\documentclass[11pt]{article}

\usepackage[margin=1in]{geometry}
\usepackage{amsmath,amsfonts,amsthm,amssymb}
\usepackage{latexsym,xspace}
\usepackage{cite}
\usepackage{engord}
\usepackage{enumitem}

\newtheorem{thm}{Theorem}
\newtheorem{cor}{Corollary}
\newtheorem{lem}{Lemma}

\newtheorem{obs}{Observation}
\theoremstyle{definition}
\newtheorem{defn}{Definition}
\theoremstyle{remark}
\newtheorem{rem}{Remark}
\newcommand{\subge}{\ensuremath{\scriptscriptstyle\geq}}
\newcommand{\subgt}{\ensuremath{\scriptscriptstyle>}}
\newcommand{\nats}{\ensuremath{\mathbb{N}}\xspace}
\newcommand{\ints}{\ensuremath{\mathbb{Z}}\xspace}
\newcommand{\nnints}{\ensuremath{\mathbb{Z}_{\subge}}\xspace}
\newcommand{\rats}{\ensuremath{\mathbb{Q}}\xspace}
\newcommand{\nnrats}{\ensuremath{\mathbb{Q}_{\subge}}\xspace}
\newcommand{\ralgs}{\ensuremath{\overline{\mathbb{Q}}}\xspace}
\newcommand{\nnralgs}{\ensuremath{{\ralgs_{\subge}}}\xspace}
\newcommand{\pralgs}{\ensuremath{{\ralgs_{\subgt}}}\xspace}
\newcommand{\algs}{\ensuremath{\mathbb{A}}\xspace}

\newcommand{\ncsp}{\ensuremath{\mathsf{\#CSP}}\xspace}

\newcommand{\ncspf}[1][\F]{\ensuremath{\mathsf{\#CSP}(#1)}\xspace}
\newcommand{\ncspg}[1][\Gamma]{\ensuremath{\ncsp(#1)}\xspace}
\newcommand{\ncspqw}[2][q]{\ensuremath{\mathsf{\#CSP}_{#1}[#2]}\xspace}
\newcommand{\ncspw}[1][\K]{\ensuremath{\mathsf{\#CSP}[#1]}\xspace}
\newcommand{\funr}[2][D]{\ensuremath{\mathfrak{F}_r(#1,#2)}\xspace}
\newcommand{\fun}[2][D]{\ensuremath{\mathfrak{F}(#1,#2)}\xspace}
\newcommand{\nump}{\textsf{\textup{\#P}}\xspace}
\newcommand{\numpc}{\nump-complete\xspace}
\newcommand{\fp}{\textsf{\textup{FP}}\xspace}
\newcommand{\fpnump}{\ensuremath{\fp^\nump}\xspace}

\newcommand{\wred}{\;{\leq}_\mathrm{w}\onept}
\newcommand{\wequ}{\;{\equiv}_\mathrm{w}\onept}

\newcommand{\CB}{\mathcal{B}}
\newcommand{\C}{\mathcal{C}}
\newcommand{\D}{\mathcal{D}}
\newcommand{\F}{\mathcal{F}}

\newcommand{\K}{\mathbb{K}}
\newcommand{\B}{\mathbb{B}}

\newcommand{\Gamm}{\ensuremath{\Gamma}\xspace}

\newcommand{\set}[1]{\left\{#1\right\}}
\newcommand{\card}[1]{\left|#1\right|}
\newcommand{\cons}[1]{\left\langle #1\right\rangle}
\newcommand{\wt}{\mathsf{w}}
\newcommand{\ba}{\mathbf{a}}
\newcommand{\bb}{\mathbf{b}}
\newcommand{\bu}{\mathbf{u}}
\newcommand{\bv}{\mathbf{v}}

\newcommand{\onept}{\hspace{1pt}}
\newcommand{\halfpt}{\hspace{0.5pt}}

\title{The complexity of weighted and unweighted \ncsp\thanks{Research supported by an EPSRC grant ``The complexity of counting in constraint satisfaction problems'' (Dyer, Goldberg, Jalsenius, Jerrum and Richerby), and by an NSERC Discovery Grant (Bulatov).}}
\author{%
    Andrei Bulatov,%
    \thanks{School of Computing Science, Simon Fraser University,
    University Drive, Burnaby, Canada, V5A 1S6.}\hspace{2mm}
    Martin Dyer,%
    \thanks{School of Computing, University of Leeds, Leeds, LS2~9JT,
    UK.}\hspace{2mm}
    Leslie Ann Goldberg,%
    \thanks{Department of Computer Science, University of Liverpool,
    Liverpool, L69~3BX, UK.}\\[1ex]
    Markus Jalsenius,%
    \thanks{Department of Computer Science,
    University of Bristol, Bristol, BS8 1UB, UK.}\hspace{2mm}
    Mark Jerrum%
    \thanks{School of Mathematical Sciences,
    Queen Mary, University of London,
    Mile End Road, London E1 4NS, UK.}\hspace{2mm}
    and David Richerby%
    \footnotemark[3]
}

\date{November 24, 2011}

\begin{document}
\maketitle
\begin{abstract}
We give some reductions among problems in (nonnegative) weighted
\ncsp which restrict the class of functions that needs to be considered in
computational complexity studies. Our reductions can be applied to both exact and approximate computation. In particular, we show that the recent dichotomy for unweighted \ncsp can be extended to rational-weighted \ncsp.
\end{abstract}

% ----------------------------------------------------------------
\allowdisplaybreaks%\linenumbers
\section{Introduction}\label{sec:introduction}
The counting complexity of the weighted constraint satisfaction
problem, for both exact and approximate computation, has been an active
research area for several years. See, for example,
\cite{Bulato08,BulDal07,CaiChe09,CaChLu10,CaChLu11,CaLuXi09,CreHer96,BuDGJR09,DyGoGJ04,%
DyGoJe09,DyGoJe10,DyGoJe10a,DyGoJR10,DyGoPa07,DyeGre00,DyeRic10,DyeRic11,GoGrJT09}.
The objective is to give a precise categorisation of the
computational complexity of problems in a given class. Easily the
most significant development in this stream of research was a recent
result of Bulatov~\cite{Bulato08}. This establishes a dichotomy for
exact counting in the whole of (unweighted) \ncsp. The
dichotomy is between problems in \fp and problems which are \numpc.
Dyer and Richerby~\cite{DyeRic10} have given an easier proof of this theorem, and have shown it to be decidable~\cite{DyeRic11}.

In this paper, we study equivalences among problems in weighted \ncsp. These
equivalences can greatly simplify the classes of problems which need to be
considered in studies of computational complexity. A particular consequence
of these results is that the dichotomy for unweighted \ncsp can be extended
to nonnegative rational-weighted \ncsp. In the results we present here, the
weights will usually lie in some subset of the nonnegative algebraic numbers, since the proofs do not appear to extend to negative weights~\cite{GoGrJT09}
or complex weights~\cite{CaChLu10}. Neither do we consider general real
numbers, since we want our results to apply to standard models of computation
and their complexity classes. An extension to a suitable model of real number
computation may be possible, though statements about complexity would need to
be modified appropriately.

The plan of the paper is as follows. In Section~\ref{sec:weightedcsp}
we define the weighted constraint satisfaction problem and establish
some notation. In Section~\ref{sec:reduction}, we define a notion of
reducibility, which we call \emph{weighted reduction}, that is used
in all our proofs. Its advantage is that the same reductions apply to
both exact and approximate computation. Section~\ref{sec:equivalence}
proves the equivalence of unweighted and rational-weighted \ncsp.
Section~\ref{sec:single} shows that a weighted \ncsp problem
can be assumed to have only one function, while retaining several
useful restrictions on instances. Finally, in
Section~\ref{sec:binary}, we show that any rational-weighted problem
is computationally equivalent to an unweighted problem with only
binary constraints. Thus any \ncsp problem is equivalent to a
canonical digraph-labelling problem. This gives another proof of the
equivalence of unweighted and rational-weighted \ncsp.

\subsection{Weighted constraint satisfaction}\label{sec:weightedcsp}

Let \ints, \rats, \ralgs and \algs denote the integers, rational
numbers, real algebraic numbers,  and (complex) algebraic numbers, respectively.
Let \nnints, \nnrats and \nnralgs denote the \emph{nonnegative}
numbers in \ints, \rats and \ralgs, respectively.  The \emph{positive}
integers $\nnints\setminus\set{0}$ will be denoted by \nats, and the positive
algebraic numbers $\nnralgs\setminus\set{0}$ by $\pralgs$.
Also $\B$ will denote $\set{0,1}$ and, if $n\in\nats$,
then $[n]$ will denote $\set{1,2,\ldots,n}$.

Let $D=\set{0,1,\ldots,q-1}$ ($q\in \nats$), which we call the
\emph{domain}, and $\K\subseteq\algs$, which we call the \emph{codomain}. Let
\[\funr{\K}\ =\ \set{f\colon D^r\to\K},\qquad\fun{\K}\ =\
\bigcup_{r\geq 1}\funr{\K}\,,\]
denote the sets of functions of all \emph{arities} from $D$ to $\K$.
We will write $r=r(f)$ for the arity of $f\in\fun{\K}$. If $r(f)=1$,
$f$ is called a \emph{unary} function and, if $r(f)=2$, it is a
\emph{binary} function.

A problem $\ncspf$ is parameterised by a finite set  $\F\subset \fun{\K}$ for some $D$ and $\K$.  An \emph{instance} $I$ of $\ncspf$ consists of a finite set of \emph{variables}~$V$ and a finite set of \emph{constraints}~$\C$. A constraint $\kappa=\cons{\bv_\kappa,f_\kappa}\in \C$ consists of a function $f_\kappa\in \mathcal{F}$ (of arity~$r_\kappa=r(f_\kappa)$) and a \emph{scope}, a sequence $\bv_\kappa=(v_{\kappa,1},\ldots,v_{\kappa,r_\kappa})$ of variables from~$V\!$, which need not be distinct. A \emph{configuration} $\sigma$ for the instance~$I$ is a function $\sigma\colon V\to D$. If $\bv=(v_1,\ldots,v_r)$, we will write $\sigma(\bv)$ for $\big(\sigma(v_1),\ldots,\sigma(v_r)\big)$. The \emph{weight} of the configuration $\sigma$ is given by
\[\wt(\sigma)=\prod_{\kappa\in \C}
f_\kappa\big(\sigma(\bv_\kappa)\big)\,.\]

Finally, the \emph{partition function} $Z_\F(I)$ is given, for
an instance $I$,  by
\[
 Z_\F(I)= \!\!\!\! \sum_{\sigma\colon V\to D} \!\!\!\! \wt(\sigma)\,.
\]
Then \ncspf denotes the problem of computing the function $Z_\F$. We
will write
\[
    \ncspqw{\K} = \set{\ncspf:\F\subset\fun{\K},\,
|D|=q},\ \ \ncspw = \bigcup_{q=2}^\infty\ncspqw{\K}\,.
\]
The case $q=1$ is clearly trivial, so we omit it from the definition
of \ncspw. The case $q=2$ is called \emph{Boolean} \ncspw.

If $\Gamm$ is a set of \emph{relations}, as in~\cite{Bulato08,BulDal07}, we
regard it as a set of functions $\F(\Gamm)\subset\fun{\B}$, so \ncsp means
\ncspw[\B]. If $R\in\Gamm$ is $r$-ary, we define $f(R)\in\F$ so that, for each $\ba\in D^r\!$, $f(\ba)=1$ if $\ba\in R$, and otherwise $f(\ba)=0$. Then we write $\ncsp(\Gamm)$ rather than $\ncsp(\F(\Gamm))$, and $Z_\Gamm$ rather than $Z_{\F(\Gamm)}$.

We consider here only \emph{non-uniform} \ncsp, where $D$ and $\F$
are considered to be objects of constant size. Thus it is only the variable
set $V\!$, and the constraint set $\C$, that determine the size of an instance.

Various other restrictions on \ncspw[\K] have been considered in the
literature, often in combination. For example, we may insist that $|\F|=m$,
for some $m\in\nats$, particularly $m=1$, e.g.~\cite{CaChLu10}. We may insist
that no function has arity greater than $r$, for some $r\in\nats$,
particularly $r=2$, e.g.~\cite{CaiChe09}. We may insist that no variable
occurs more than $k$ times in an instance, e.g.~\cite{DyGoJR10}.  We may
insist that the functions in $\F$ possess some particular property, such as
symmetry, e.g.~\cite{DyGoJe10a}. We do not consider these
restrictions in any detail here. However, we will make use of the following restricted version of \ncspw[\K] in Section~\ref{sec:binary}.

A unary function which must be applied \emph{exactly once} to each variable
$v\in V$ will be called a \emph{vertex weighting}, and its function values
\emph{vertex weights}. Thus, if $\lambda\colon D\to \K$ is a vertex weighting, any
instance $I$ must contain exactly one constraint of the form
$\cons{(v),\lambda}$ for each $v\in V\!$. Observe that it is not necessary to allow multiple vertex weightings $\lambda_1$, $\lambda_2$, \ldots, $\lambda_m$, since these can be combined into one equivalent vertex weighting $\lambda=\lambda_1\lambda_2\cdots\lambda_m$.

Our definition of vertex weights conforms to the use of
similar terminology elsewhere, for example in~\cite{DyeGre00}. We will denote
the problem with $\F\subset\fun{\K}$ and vertex weighting $\lambda\colon D\to\K$ by
\ncspf[\F;\lambda]. The problem \ncspf[\F;\lambda] is a restriction on the inputs to an associated \ncspw[\K] problem, \ncspf[\F\cup\set{\lambda}]. In an instance of \ncspf[\F\cup\set{\lambda}], $\cons{(v),\lambda}$ can appear any number of times, including zero, for each $v\in V$; in an instance of \ncspf[\F;\lambda], each $\cons{(v),\lambda}$ appears precisely once.

We will also consider \emph{approximate} evaluation of $Z_\F$, meaning
\emph{relative} approximation. Thus, given $\epsilon>0$ we wish to compute an estimate $\widehat{Z}_\F(I)$ of $Z_\F(I)$, for all $I$, such that
\begin{equation}\label{eq10}
  |\widehat{Z}_\F(I)-Z_\F(I)|\ \leq\ \epsilon\halfpt|Z_\F(I)|\,.
\end{equation}
For \emph{randomised} approximation, we require only that this holds
with sufficient probability. See~\cite{DyGoGJ04}, for example, for
further details. Observe that definition~\eqref{eq10} applies equally
if $Z_\F$ can take negative or complex values, though we consider only
nonnegative real weights, here.

\subsection{Weighted reductions}\label{sec:reduction}
Let $\Sigma$ be a finite alphabet, and let $F\colon\Sigma^*\to\algs$.
We are interested in evaluating $F$ only for strings $x$ that
encode instances $I$ of some computational problem. However, we will make
$F$ into a total function by setting $F(x)=0$ if $x\in\Sigma^*$ does not
encode an instance. In particular, $F(\varepsilon)=0$ for the empty
string $\varepsilon$.
\begin{defn}
Let $F_1,F_2\colon\Sigma^*\to\algs$. A \emph{weighted} reduction from
$F_1$ to $F_2$ is a pair of \fp-computable functions
$\phi\colon\Sigma^*\to\pralgs$,\, $\psi\colon\Sigma^*\to\Sigma^*$ such that
$F_1(x)=\phi(x)F_2(\psi(x))$ for all $x\in\Sigma^*\!$.
\end{defn}
In constructing a weighted reduction, we can clearly restrict
attention to strings $x$ that encode instances. Otherwise, we will simply
take $\phi(x)=1$, and $\psi(x)=\varepsilon$, the empty string.

Weighted reductions generalise the ``\emph{simulates}'' concept
defined in~\cite{DyGoJe09}. Parsimonious reductions~\cite{Simon75} are
contained as the special case $\phi(x)=1$ for all $x\in\Sigma^*\!$. Weighted
reduction relaxes the definition of parsimonious reduction by allowing a
positive ``weight'' $\phi(x)$ for each $x\in\Sigma^*\!$. The generalisation is
valuable in two respects. First, it preserves relative approximation of the
functions $F_1$ and $F_2$ and, hence retains the most useful property of
parsimonious reductions. If $\widehat{F}_2(x)$ is an approximation to $F_2(x)$ with relative error $\epsilon$, it follows easily that $\widehat{F}_1(x)=\phi(x)\widehat{F}_2(\psi(x))$ is an approximation to $F_1(x)$ with relative error $\epsilon$. Weighted reduction is, in fact, a simple type of \emph{AP-reduction}, as defined in~\cite{DyGoGJ04}.

Second, weighted reductions allow us to relax the cumbersome condition $F_1,
F_2\to\nnints$, required by parsimonious reductions, so we can work with the
natural classes of functions. All reductions used in this paper will be
weighted reductions.

We write $F_1\wred F_2$ to indicate the existence of a weighted reduction
from $F_1$ to $F_2$. If $F_1\wred F_2$ and $F_2\wred F_1$, we say that the
functions are \emph{equivalent} (under weighted reductions), and we write
$F_1\wequ F_2$. Thus, if $F_1\wequ F_2$, then $F_1$ and $F_2$ will have the same computational complexity for both exact and approximate
computation. This would not be true for approximate computation if we were to use the weaker notion of \emph{Turing} reducibility, as is usual for exact computation in the class \nump~\cite{Valian79}.

If $\F_1$ and $\F_2$ are two classes of functions such that, for all
$F_1\in\F_1$, there is an $F_2\in \F_2$ such that $F_1\wequ F_2$, and
conversely, for all $F_2\in\F_2$, there is an $F_1\in \F_1$ such that
$F_2\wequ F_1$, we will write $\F_1\wequ \F_2$.

The reason for making
this definition in terms of equivalence, rather than reduction, is
that, when $\F_1$ has a classification into functions of different
complexity, for example a \emph{dichotomy}, then this classification
is inherited by any $\F_2\wequ \F_1$. In our proofs below, we will
always have $\F_2\subset\F_1$, so proving that $\F_1\wequ \F_2$ will only
require showing that, for all $F_1\in\F_1$, there is an $F_2\in \F_2$
such that $F_1\wequ F_2$.

\section{Equivalence of \ncspw[\nnrats] and \ncsp}\label{sec:equivalence}

Under weighted reductions, we may assume that all instances of \ncspf have
every $v\in V$ appearing in the scope of some constraint. Otherwise, suppose
the variables in $V_0\subseteq V$ do not appear in the instance $I$, and let
$n_0=|V_0|$. Let $I'$ be identical to $I$ except that $V'=V\setminus V_0$.
Then $Z_\F(I)=|D|^{n_0} Z_\F(I')$, so there is an equivalent problem of the
required type using the reversible reduction $\phi(I)=|D|^{n_0}$ and
$\psi(I)=I'\!$. We will assume that this has been done, so all variables in $V$
appear in the scope of some constraint in $\C$.

Observe also that \emph{repeated} constraints are irrelevant in \ncsp, but not
in \ncspw when $\K\neq\B$. We may assume that instances of \ncspg do not have
repeated constraints, since otherwise there is trivial equivalence with this
case.

First, suppose $\F\subset\fun{\rats}$. Then, by computing a common
denominator $N\in\nats$ for the ranges of the functions in $\F\!$, we
can write $f'(\ba)=Nf(\ba)$, for each $f\in\F$ and we have
$f'(\ba)\in\ints$ for all $\ba\in D^{r(f)}$.  Let $\F'=\set{f':
  f\in\F}$.

\begin{lem}\label{lem05}
   If $\F'$ is obtained from $\F$ as above, then $\ncspf\wequ\ncspf[\F']$.
\end{lem}
\begin{proof}
If $I$ is an instance of \ncspf, and $I'$ is the
corresponding instance of \ncspf[\F'], we have
$Z_\F(I)=N^{-k}Z_{\F'}(I')$, where $k=|\C(I)|$.
Thus, letting $\phi(I)=N^{-k}$ and $\psi(I)=I'\!$, there
is a weighted reduction from $Z_\F(I)$ to $Z_{\F'}(I')$, and hence
$\ncspf\wred\ncspf[\F']$. Reversing this reduction gives
$\ncspf\wequ\ncspf[\F']$.
\end{proof}
\begin{cor}
$\ncspw[\rats]\wequ \ncspw[\ints]$.
\end{cor}
\begin{proof}
Since $\ints\subset\rats$, this follows immediately from Lemma~\ref{lem05}.
\end{proof}
Now, given $\F\subset\fun{\nnints}$, we will construct a set of relations
$\Gamm(\F)$, with domain $A$, as follows. For each function $f\in\F$ and each
$\ba\in D^r\!$, where $r=r(f)$, we create a set $\D_{f,\ba}$ of cardinality
$f(\ba)$ such that these sets are all mutually disjoint, and also disjoint
from $D$. Let
\[A=D\ \cup\,\bigcup_{f\in\F}\bigcup_{\ba\in D^r}
\D_{f,\ba},\quad\mathrm{so}\quad
|A|=|D|+\sum_{f\in\F}\sum_{\ba\in D^r}f(\ba)\,.\]
Now, for every $f\in\F\!$, we construct a relation $R(f)\subseteq A^{r+1}\!$,
as follows. For each $r$-tuple $\ba$ with $f(\ba)>0$, we create an
$(r+1)$-tuple $(\ba,w)\in R(f)$ for every $w\in\D_{f,\ba}$.
\begin{obs}\label{obs10}
All tuples $(\ba,w)\in R(f)$ have $\ba\in D^r$ and $w\notin D$.
\end{obs}
\begin{obs}\label{obs20}
For each $w\in A\setminus D$, there is a unique $f\in\F$ and $\ba\in
D^r$ such that $(\ba,w)\in R(f)$.
\end{obs}
We use these observations to prove the following equivalence.
\begin{lem}\label{lem10}
If $\Gamm=\Gamm(\F)$, as defined as above, then $\ncspf\wequ\ncspg$.
\end{lem}
\begin{proof}
Suppose $I$ is an instance of $\ncspf$. For each constraint
$\kappa=\cons{\bv_\kappa,f_\kappa}$, we create the constraint
$\kappa'=\cons{(\bv_\kappa,v_\kappa),R(f_\kappa)}$ in an instance
$I'$ of \ncspg, where $v_\kappa$ is a new variable.
Thus $I'$ has variable set $V'=V\cup\set{v_\kappa:\kappa\in\C}$.
Now, each configuration $\sigma\colon V\to D$ in $I$ can be identified
with the set of configurations
$\sigma'\colon V'\to A$ in $I'$ that agree with $\sigma$ over $V\!$.
Thus $\sigma'(\bv_\kappa)=\sigma(\bv_\kappa)$
and $\sigma'(v_\kappa)\in \D_{f_\kappa,\sigma(\bv_\kappa)}$. By
Observation~\ref{obs20}, these partition the set of all $\sigma'$
having nonzero weight. Since there are exactly $f(\sigma(\bv_\kappa))$
choices for $\sigma'(v_\kappa)$, we have $Z_\F(I)=Z_\Gamm(I')$.
We take $\phi(I)=1$, $\psi(I)=I'$ and hence we have $\ncspf\wred\ncspg$.

Conversely, let $I$ be an instance of \ncspg, and let
$\kappa=\cons{(\bv,v),R(f)}$ be any constraint. If $v$ also appears in the
tuple $\bv'$ of a constraint $\kappa'=\cons{(\bv'\!,v'),R(f')}$, then there can
be no configuration $\sigma\colon V\to A$ with nonzero weight, by Observation~\ref{obs10}. Thus
$Z_\Gamm(I)=0$, so we take $\phi(I)=1$ and $\psi(I)=\varepsilon$. Now, if $v$
appears other than in constraint $\kappa$, it must be in a constraint
$\kappa'=\cons{(\bv'\!,v),R(f')}$. But then, from Observation~\ref{obs20}, any
$\sigma\colon V\to A$ has nonzero weight only if $\sigma(\bv')=\sigma(\bv)$ and
$f'=f$. Thus we may add the equalities $\bv'=\bv$ and delete the constraint
$\kappa'\!$. Repeating this procedure, we construct an instance $I_0$ of \ncspg
such that $Z_\Gamm(I_0)=Z_\Gamm(I)$, and each constraint
$\kappa=\cons{(\bv,v),R(f)}$ in the constraint set $\C_0$ of the instance
$I_0$ has a unique variable $v=v_\kappa$. Thus $I_0$ is precisely the
instance of \ncspg which would result from applying the construction in the
first part of the proof to the instance $I'_0$ of \ncspf with variables
$V'=V\setminus\set{v_\kappa:\kappa\in\C_0}$ and constraints
$\cons{\bv_\kappa,f_\kappa}$ ($\kappa\in\C_0$). It follows that
$Z_\Gamm(I)=Z_\Gamm(I_0)=Z_\F(I')$. So we may take $\phi(I)=1$, $\psi(I)=I'$
and hence we have $\ncspg\wred\ncspf$.
\end{proof}
\begin{rem}\label{rem10}
The reader will note that the size of the resulting unweighted
problem increases dramatically with the size of the weights. Since these
weights are constants in the non-uniform model, this has no impact on the
complexity. However, we make no claims for the practicality of the reduction.
\end{rem}
\begin{thm}\label{thm10}
$\ncspw[\nnrats]\wequ\ncsp$.
\end{thm}
\begin{proof}
This follows directly from $\B\subset\nnrats$ and Lemma~\ref{lem10}.
\end{proof}
As noted above, Bulatov~\cite{Bulato08} has shown a dichotomy for
\ncsp into problems which are in \fp and problems which are \numpc
(see also~\cite{DyeRic10}). Combining this with Theorem~\ref{thm10},
and an argument given in  Section~1.3 of~\cite{DyGoJe09},
we have the following.
\begin{thm}[{\rm Dichotomy}]\label{thm40}
Any problem in \ncspw[\nnrats] is either in \fp or is complete for
$\fpnump\!$.
\end{thm}
\begin{rem}\label{rem20}
The method of proof used in Theorem~\ref{thm10} clearly fails for irrational weights. However, since this paper was written, Cai, Chen and Lu~\cite{CaChLu11} have proved a general dichotomy theorem for weights in \nnralgs.

\end{rem}
\begin{rem}\label{rem25}
Theorem~\ref{thm10} may have analogues for mixed-sign and complex
weights. However, the above method of proof encounters technical problems
with repeated constraints in these cases.
\end{rem}

\section{Reduction to a single function}\label{sec:single}
Here we consider $\F\subset\fun{\nnralgs}$.
We will show that \ncspf is equivalent to \ncspf[\set{g}] for a single
function $g\in\fun{\nnralgs}$. We abbreviate \ncspf[\set{g}] to \ncspg[g].

We may assume that no $f\in \F$ is identically zero.  Otherwise,
if $f(\ba) = 0$ for all $\ba\in D^{r(f)}\!$, then $Z_\F(I)=0$ for any
instance $I$ of \ncspf where $f$ appears in a constraint.  Then, letting $\F'=\F\setminus\{f\}$, $Z_\F(I) = Z_{\F'}(I)$, we have $\ncspf\wred\ncspf[\F']$.

Let
\[M(f) = \sum_{\ba\in D^r} f(\ba)\ >\ 0\qquad(f\in\F)\,.\]
Now, let $\ell=|\F|$, and let $\F=\set{g_1,g_2,\ldots,g_\ell}$,
$r_j=r(g_j)$ and $M_j=M(g_j)$ ($j\in[\ell]$). Let $s=\sum_{j=1}^\ell
r_j$ and define $g\colon D^s\to\nnralgs$ by
\[g(\ba_1,\ba_2,\ldots,\ba_\ell)=\prod_{j=1}^\ell g_j(\ba_j)
\qquad(\ba_j\in D^{r_j};\,j\in[\ell])\,.\]
If $\kappa=\cons{\bv_\kappa,f_\kappa}$ is a constraint of an
instance $I$ of \ncspf, let $i_\kappa$ be defined by $i_\kappa=j$ if
$f_\kappa=g_j$.
\begin{thm}\label{thm50}
For all $\F\subset\fun{\nnralgs}$, there exists $g\in\fun{\nnralgs}$
such that $\ncspf\wequ\ncspf[g]$.
\end{thm}
\begin{proof}
The required $g$ is the function $g(\F)$ constructed above.
For any instance $I$ of \ncspf, construct an instance $I'=\psi(I)$
of $\ncspf[g]$ by padding each constraint $\kappa=\cons{\bv_\kappa,f_\kappa}$
that has $f_\kappa=g_{i_\kappa}$, to give
\[\kappa'\ =\
\cons{\big(\bu_{1,\kappa},\ldots,\bu_{i_\kappa-1,\kappa},\bv_\kappa,
\bu_{i_\kappa+1,\kappa},\ldots,\bu_{\ell,\kappa}\big) , g }\,,\]
where $\bu_{j,\kappa}$ $( j \in[\ell],j\neq i_\kappa)$ is an $r_j$-tuple of
new variables not in $V\!$, and disjoint for each $j\neq i_\kappa$ and
$\kappa\in\C$. Thus $I'$ has variable set $V'\!$, with
\[ |V'|\ =\ |V|+\sum_{\kappa\in\C}\,\sum_{j\neq i_\kappa} r_j\ \leq\
s|\C|\,.\]
Any $\sigma'\colon V'\to D$ decomposes into $\sigma\colon V\to D$ and
$\sigma_{j,\kappa}\colon\,\bu_{j,\kappa}\to D$ ($j\neq i_\kappa,
\kappa\in\C$). Clearly, for each value of $j$ and $\kappa$,
\[ \sum_{\sigma_{j,\kappa}}\, g_j\big(\sigma_{j,\kappa}(\bu_j)\big)\
=\ \sum_{\ba\in D^{r_j}} g_j(\ba)\ =\ M_j\ >\ 0\,.\]
Thus, it follows that
\[Z_g(I')\ =\  \chi(I)Z_\F(I),\quad\mathrm{where}\ \  \chi(I)\,
=\,\prod_{\kappa\in\C}\,\prod_{j\neq i_\kappa} M_j\,>\, 0\,,\]
which gives a weighted reduction from \ncspf to  $\ncspf[g]$
with $\phi(I)=1/\chi(I)$.

In the other direction, the reduction is straightforward. Suppose
$\kappa=\cons{\bv_\kappa,g}$ is a constraint of an arbitrary
instance $I$ of \ncspf[g], where
$\bv_\kappa = (\bv_{1,\kappa},\ldots,\bv_{\ell,\kappa})$,
with $\bv_{j,\kappa}\in V^{r_j}$ ($j\in[\ell]$). Create the instance
$I'=\psi(I)$ of \ncspf with constraints
$\C'=\set{\cons{\bv_{j,\kappa},f_j}:
j\in[\ell],\, \kappa\in\C}$.
Clearly, $Z_g(I)=Z_\F(I')$, so we have a
weighted reduction from $\ncspf[g]$ to \ncspf with $\phi(I)=1$.
\end{proof}
\begin{rem}\label{rem30}
Theorem~\ref{thm50} does not appear to carry over to negative or
complex weights. The proof above fails because we may have $M(f)=0$,
so $\chi(I)=0$, and hence $\phi(I)$ will be undefined.
\end{rem}
The important features of the equivalence of Theorem~\ref{thm50} are
\begin{enumerate}
  \item[(i)] it does not change the domain $D$;
  \item[(ii)] it preserves approximation, since the reductions are weighted;
  \item[(iii)] it preserves relations, since $\B$ is closed under product;
  \item[(iv)] it preserves the maximum number of occurrences (\emph{degree})
  of variables.
\end{enumerate}
Thus, for most complexity studies, allowing multiple functions or
relations in \ncsp does not increase generality. Theorem~\ref{thm50}
can be used to simplify proofs given, for example,
in~\cite{Bulato08,CreHer96,DyGoJe09,DyGoJe10,DyGoJR10,DyeRic10}.

\section{Reduction to binary constraints}\label{sec:binary}
The proof of equivalence of  \ncspw[\nnrats] and \ncsp in
Section~\ref{sec:equivalence} is probably the simplest, but not the only
construction. We present a different proof here, which is of interest in its
own right. An instance of \ncspf is reduced to an instance \ncspg, where
\Gamm is a set of \emph{binary} relations. Thus any problem in \ncsp can be
stated as an equivalent problem concerning \emph{digraphs}.

We give the proof in two parts. In the first part, we show equivalence of any
problem in \ncspw[\nnralgs] with a problem having a vertex weighting and a
set of binary \emph{relations}. We will then show that this vertex-weighted
problem is equivalent to an unweighted digraph problem.

\begin{thm}\label{thm60}
If $\F$ is a finite subset of $\fun{\nnralgs}$, then $\ncspf \wequ\ncspf[\CB;\lambda]$, where
$\CB$ is a finite set of binary relations and $\lambda\colon D\to\nnralgs$ is a vertex weighting.
\end{thm}
\begin{proof}
We may assume, by Theorem~\ref{thm50}, that $\F=\set{g}$ with
$g\in\funr{\nnralgs}$, for some $r$. Thus, to specify a constraint,
we need only give its scope. We also assume that every variable
appears in some scope, as discussed in Section~\ref{sec:equivalence}.
Then \ncspf[\CB;\lambda] is specified as follows.
\begin{enumerate}
  \item[(a)] The domain $A=D^r\!$, so $\ba=(a_1,a_2,\ldots,a_r)\in A$ for all $a_1,a_2,\ldots,a_r\in~D$.
  \item[(b)] For all $\ba\in A$, $\lambda(\ba)=g(\ba)$.
  \item[(c)] For each $i,k\in[r]$, there is a $\beta_{ik}\in\CB$ such
  that for all $\ba,\bb\in A$,
  \[\beta_{ik}(\ba,\bb)\ =\ \left\{
      \begin{array}{ll}
        1, & \hbox{if $\ba_i=\bb_k$;} \\
        0, & \hbox{otherwise.}
      \end{array}
    \right.\]
\end{enumerate}
Let $I$ be any instance of $\ncspf[g]$, with variable set $V$ and
constraint set $\C$. We can construct an equivalence relation $\sim$ on
$\C\times[r]$ such that $(\iota,i)\sim(\kappa,k)$ if, and only if,
$\bv_{\iota,i}$ and $\bv_{\kappa,k}$ are the same variable $v\in V\!$.
Thus $\sim$ has $|V|$ equivalence classes, each class corresponding
to a variable in $V\!$.

We now construct an instance $I'=\psi(I)$ of \ncspf[\CB;\lambda], which
has variable set $V'$ and constraint set $\C'\!$,  as follows.
\begin{enumerate}
\item[(i)] For each $\kappa\in\C$, we have a variable $\kappa\in V'\!$. Thus $V'=\C$.
\item[(ii)] For all $\kappa\in V'\!$, we have one constraint $\cons{(\kappa),\lambda}\in \C'\!$.
Thus $\lambda$ is a vertex weighting.
\item[(iii)] For all $\iota,\kappa\in\C$, we have a constraint
$\cons{(\iota,\kappa),\beta_{ik}}\in \C'$ for each $i$, $k$ with $(\iota,i)\sim(\kappa,k)$.
\end{enumerate}

Let $\sigma\colon V\to D$ be any configuration for $I$. Then $(\iota,i)\sim(\kappa,k)$
implies $\beta_{ik}(\sigma(\bv_\iota)$, $\sigma(\bv_\kappa))=1$. In
turn, this implies $\sigma(\bv_{\iota,i})=\sigma(\bv_{\kappa,k})$,
as is required by the variables $\bv_{\iota,i}$ and $\bv_{\kappa,k}$
being identical. Thus there is a bijection between the configurations
$\sigma$ of $I$ having nonzero weight and the configurations $\sigma'$ of $I'$
having nonzero weight. Let us write $\sigma'=\xi(\sigma)$ for this bijection.
Note that $\sigma'=\xi(\sigma)$ then satisfies
$\beta_{ik}(\sigma'(\iota),\sigma'(\kappa)\big)
=\beta_{ik}(\sigma(\bv_\iota),\sigma(\bv_\kappa)\big)=1$
if $(\iota,i)\sim(\kappa,k)$. Thus, with $\sigma'=\xi(\sigma)$, we have
\[
    \wt(\sigma')\ =\ \prod_{\kappa\in \C}
    \lambda\big(\sigma'(\kappa)\big)\prod_{(\iota,i)\sim(\kappa,k)}
    \beta_{ik}(\sigma'(\iota),\sigma'(\kappa)\big)
     \ =\ \prod_{\kappa\in \C} g\big(\sigma(\bv_\kappa)\big)\ =\ \wt(\sigma),
\]
so the bijection $\xi$ is weight-preserving. Thus $Z_{\CB;\lambda}(I)=Z_g(I')$,
and we have a weighted reduction from \ncspf[g] to \ncspf[\CB;\lambda], with
$\phi(I)=1$.

Conversely, suppose $I$ is any instance of \ncspf[\CB;\lambda] with variable
set $V$ and constraint set $\C$. We construct an instance $I'=\psi(I)$ of
\ncspf[g], with variable set $V'$ and constraint set $\C'\!$, as follows. Note
that $A$ and the $\beta_{ik}$ are not arbitrary, but have been derived as
in~(a) and~(c) above. Thus, in particular, we can easily
deduce the value of $r$. We now create a relation $\sim$
on the set $V^*=V\times[r]$, as follows. For ease of notation, we
will write $(u,i)\in V^*$ as $u_i$. For $u,v\in V\!$, let $u_i\sim v_k$ if there is a constraint $\cons{(u,v),\beta_{ik}}\in\C$.

Now, suppose $\sigma$ is any configuration of $I$. Then, for any $v\in
V\!$, we have $\sigma(v)=(a_1,\ldots,a_r)\in D^r\!$, from~(a) above. Let
us write $\sigma_i(v)=a_i$ ($i\in[r]$). Now, define $\sigma'\colon V^*\to D$ from
$\sigma$ by $\sigma'(v_i)=\sigma_i(v)$ for all $v\in V^*\!$, $i\in[r]$.
We will write $\sigma'=\zeta(\sigma)$ for this function.
If, for any $u,v\in V\!$, we have $u_i\sim v_k$, then
we must have $\beta_{ik}(\sigma(u),\sigma(v))=1$. From~(c) above,
this implies that $\sigma_i(u)=\sigma_k(v)$, and hence
$\sigma'(u_i)=\sigma'(v_k)$, where $\sigma'=\zeta(\sigma)$. Thus we
can extend the relation $\sim$, as follows. We have $u_i\sim v_k$ if
$\sigma'(u_i)=\sigma'(v_k)$ for all $\sigma'=\zeta(\sigma)$, where
$\sigma$ is a configuration of $I$. Clearly, $\sim$ is now an
equivalence relation, which ``identifies'' the variables $u_i$ and
$v_k$. More precisely, the variable set $V'$ of $I'$ will be the
set of equivalence classes $V^*/\!\!\sim$\/. For any $v_k\in V^*\!$,
we write $\bar{v}_k$ for its equivalence class. Let $\sigma'\colon V^*\to D$
be such that $\sigma'=\zeta(\sigma)$ for some $\sigma\colon V\to D^r\!$. Then we
can define $\bar{\sigma}\colon V'\to D$ by $\bar{\sigma}(\bar{v}_k)=\sigma(v_k)$
for all $v_k\in \bar{v}_k$. Thus we have constructed a bijection
between the configurations $\sigma$ of $I$ having nonzero weight and the
configurations $\bar{\sigma}$ of $I'$ having nonzero weight. We will write
$\bar{\sigma}=\xi(\sigma)$ for this bijection.

Now, $I'$ will have constraint set
\[ \C'=\set{\cons{\bar{\bv},g}:
\bar{\bv}=(\bar{v}_1,\ldots,\bar{v}_r), \ v\in V}\,.\]
Then, with $\bar{\sigma}=\xi(\sigma)$, we have
\[
    \wt(\bar{\sigma})\ =\ \prod_{v\in V'}
    g\big(\bar{\sigma}(\bar{\bv})\big)
     =\ \prod_{v\in V} \lambda\big(\sigma(v)\big)\prod_{u_i\sim v_k}
     \beta_{ik}(\sigma(u),\sigma(v)\big)\ =\ \wt(\sigma),
\]
so the bijection $\xi$ is weight-preserving. Thus $Z_g(I')=Z_{\CB;\lambda}(I)$,
and we have a weighted reduction from \ncspf[\CB;\lambda] to \ncspf[g] with
$\phi(I)=1$.
\end{proof}
\begin{rem}\label{rem35}
In fact, Theorem~\ref{thm60} holds, more generally, for $\F\subset\fun{\algs}$. The proof above needs modification, however, since we cannot apply Theorem~\ref{thm50}. Instead, we use binary relations $\beta^{f,h}_{i,j}$ for each $f,h\in\F\!$, $i\in[r(f)]$, $j\in[r(h)]$, and
domain $A=\bigcup_{f\in\F}\set{\ba_f:\ba\in D^{r},\, r=r(f)}$. We omit
the details, since we currently have no application for this generalisation. The proof of Theorem~\ref{thm70} below is valid only for  $\F\subset\fun{\nnrats}$.
\end{rem}
This yields a different proof of the equivalence of
\ncspw[\nnrats] and \ncsp.
\begin{thm}\label{thm70}
Let $\CB$ be a set of binary relations, and let $\lambda\colon D\to\nnrats$ be a vertex weighting. Then $\ncspf[\CB;\lambda]\wequ\ncspg$, where $\Gamm$ is a set of binary relations.
\end{thm}
\begin{proof}
We will use the equivalence proved in Lemma~\ref{lem05}.
Thus we may take $\lambda\colon A\to \nnints$. Then we use a construction
similar to that of Section~\ref{sec:equivalence}. Note that, if
$\lambda(a)=0$ for any $a\in A$, we can delete $a$ from $A$.
All configurations with $\sigma(v)=a$ for any $v\in V$ have zero weight
and do not contribute to the partition function. Thus we may assume
$\lambda(a)>0$ for all $a\in A$. Then let
\[ B_a\ =\ \set{\onept(a,i):\, i\in[\lambda(a)]\onept}\quad(a\in A),
\qquad B\ =\ \bigcup_{a\in A} B_a\,,\]
where we will again write $(a,i)$ as $a_i$. Then $\Gamm$ will comprise a set of
binary relations $\gamma$ on the domain $B$ such that, for each $\beta\in\CB$,
there is a $\gamma(\beta)$ defined by
\[ \gamma(\beta)\ =\ \set{\halfpt(a_i,b_j): (a,b)\in\beta,\,
i\in[\lambda(a)],\, j\in[\lambda(b)]\halfpt}\,.\]
Clearly, this gives a bijection between $\CB$ and $\Gamm\!$, so we may
also write $\beta=\beta(\gamma)$.

Now, let $I$ be any instance of  \ncspf[\CB;\lambda] with variable set $V$ and
constraint set $\C$. Then $I'=\psi(I)$ will have variable set $V'=V$ and
constraint set
\[ \C'\ =\ \set{\halfpt\cons{(u,v),\gamma}:\gamma=\gamma(\beta),\,
\cons{(u,v),\beta}\in\C\halfpt}\,.\]
Let $\sigma'$ be any satisfying configuration of $I'\!$. This can be
mapped to a configuration $\sigma$ of $I$ satisfying all its binary constraints
by $\sigma(v)=a$ if $\sigma'(v)=a_i$ for some $i\in[\lambda(a)]$.
Let us write $\sigma=\eta(\sigma')$ for this function. Then,

\begin{eqnarray*}
\sum_{\sigma'\in\eta^{-1}(\sigma)}\!\!\!\wt(\sigma')&=& \card{\eta^{-1}(\sigma)}
\ =\ \Big|\prod_{v\in V}\set{\sigma'(v):\sigma'(v)\in B_{\sigma(v)}}\Big|\\
&=&\prod_{v\in V}\card{B_{\sigma(v)}}\ =\ \prod_{v\in V}
\lambda(\sigma(v))\ =\ \wt(\sigma)\,.
\end{eqnarray*}

Thus $Z_{\CB;\lambda}(I)=Z_\Gamm(I')$, so we have $\phi(I)=1$, and we have shown that
$\ncspf[\CB;\lambda]\wred\ncspg$.

Conversely, if $I$ is any instance of \ncspg with variable set $V'$
and constraint set $\C'\!$, we create an instance $I'=\psi(I)$ with
variable set $V=V'$ and constraint set
\[ \C\ =\ \set{\halfpt\cons{(u,v),\beta}:\cons{(u,v),\gamma}\in\C'\!,\,
\beta=\beta(\gamma)\halfpt}\,.\]
Reversing the above calculation yields $Z_{\CB;\lambda}(I')=Z_\Gamm(I)$, so
$\phi(I)=1$ and $\ncspg\wred\ncspf[\CB;\lambda]$. Hence $\ncspf[\CB;\lambda]\wequ\ncspg$.
\end{proof}
\begin{rem}\label{rem40}
Observe that this proof does not really require that the relations in $\CB$
are all binary. We have made this restriction only for notational simplicity,
and because it is the case needed for the following application.
\end{rem}
Combining Theorems~\ref{thm60} and~\ref{thm70}, we have an
alternative proof of the results implied by Theorem~\ref{thm10}. That is,
$\ncspw[\nnrats]\wequ\ncsp$ and a dichotomy theorem for \ncspw[\nnrats].
\begin{rem}\label{rem50}
Theorems~\ref{thm60} and~\ref{thm70} determine a canonical form for
$\ncspw[\nnrats]$. The general problem is a set of $k$ digraphs,
$H_1,H_2,\ldots,H_k$ on the same vertex set $D$. An instance is a set
of $k$ digraphs, $G_1,G_2,\ldots,G_k$ on the same vertex set $V\!$. A
satisfying configuration is a labelling of $V$ with $D$
that induces a homomorphism from $G_i$ to $H_i$ for all $i\in[k]$.
Cai and Chen~\cite{CaiChe09} have given a decidable
dichotomy theorem for the case $k=1$ of this problem.
\end{rem}
\begin{rem}\label{rem60}
The digraphs $H_1,H_2,\ldots,H_k$ in the canonical problem of
Remark~\ref{rem50} can be taken to be \emph{directed acyclic
graphs} (\emph{DAGs}), though possibly with loops. A decidable
dichotomy theorem for the case $k=1$ of this problem (without
loops) was given by Dyer, Goldberg and Paterson~\cite{DyGoPa07}.
The simplification can be justified as follows. Suppose we impose
an arbitrary linear order on $A$. By the symmetries
$\beta_{ij}(\bu,\bv)=\beta_{ji}(\bv,\bu)$ in the proof of
Theorem~\ref{thm60}, we need only include $(\bu,\bv)$ in the
relation $\beta_{ij}$ if $\bu\leq\bv$. Thus each $\beta_{ij}$
describes a DAG, perhaps having loops on its vertices.
\end{rem}
%\nolinenumbers
% ----------------------------------------------------------------

\end{document}